\documentclass[a4paper,12pt]{article}

\makeatletter

\@addtoreset{equation}{section}
\makeatother
\usepackage{amsfonts}
\usepackage{eucal}
\usepackage{amsmath}
\usepackage{amssymb}
\usepackage{mathrsfs}
\usepackage{cases}
\usepackage{amsthm}

\renewcommand{\(}{\left ( }
\renewcommand{\)}{\right ) }
\renewcommand{\H}{{\mathcal H}}

\newcommand{\E}{{\cal E}}

\renewcommand{\L}{{\cal L}}

\newcommand{\C}{\mathbb {C}}
\newcommand{\R}{\mathbb {R}}

\newcommand{\N}{\mathbb {N}}
\newcommand{\Z}{\mathbb {Z}}

\def\i<#1>{\langle #1 \rangle}
\def\l<#1>{\left\langle #1 \right\rangle}

  \makeatother

  \newtheorem{Theorem}{Theorem}[section]
  
  \newtheorem{Proposition}[Theorem]{Proposition}
  \newtheorem{Lemma}[Theorem]{Lemma}
  \newtheorem{Corollary}[Theorem]{Corollary}
  \newtheorem{Remark}[Theorem]{Remark}
  \newtheorem{Example}[Theorem]{Example}
  \newtheorem{Assumption}{Hypothesis}

\makeatletter

\def\@thesis{}
 \def\id#1{\def\@id{#1}}
 \def\department#1{\def\@department{#1}}

\def\@maketitle{
 \begin{center}
 {\large \@title \par}%
 \vspace{5mm}
 {\@author \par}%
\vspace{5mm}

\end{center}
 \par\vskip 1.5em
 }

\makeatother

\title{{\bf Localization for a one-dimensional 
split-step quantum walk \\ with bound states 
robust against perturbations}}
  \author{Toru Fuda\footnote{Department of Mathematics and Science, School of Science and Engineering, Kokushikan University, 4-28-1, Setagaya, Setagaya-Ku, Tokyo 154-8515, Japan,
  		\\E-mail: fudat@kokushikan.ac.jp},\ \ 
  	Daiju Funakawa\footnote{Department of Electronics and Information Engineering, Hokkai-Gakuen University, Sapporo 062-8605, Japan,
  		\\E-mail:funakawa@hgu.jp},\ \ 
  	Akito Suzuki\footnote{Division of Mathematics and Physics, Faculty of Engineering, Shinshu University, Wakasato, Nagano
  		\\380-8553, Japan, E-mail: akito@shinshu-u.ac.jp}}

\begin{document}
	
\maketitle
	
\begin{abstract}
For given two unitary and self-adjoint operators on a Hilbert space, 
a spectral mapping theorem was proved in \cite{HiSeSu}. 
In this paper,
as an application of the spectral mapping theorem, 
we investigate the spectrum of a one-dimensional split-step quantum walk. 
We give a criterion for when there is no eigenvalues around $\pm 1$ 
in terms of a discriminant operator. 
We also provide a criterion for when eigenvalues $\pm 1$ exist  
in terms of birth eigenspaces. 
Moreover, we prove that eigenvectors
from the birth eigenspaces decay exponentially at spatial infinity
and that the birth eigenspaces are robust against perturbations. 
\end{abstract}

\section{Introduction}
During the last two decades, 
increasing attention has been paid to discrete-time quantum walks
(see \cite{Am03, Ke07, Ko08, VA12, P13, MaWa14} and references therein), 
which are quantum counterparts of classical random walks. 
Motivated by Grover's search algorithm \cite{Gr96}, 
Szegedy \cite{Sz} quantized  a random walk on a finite bipartite graph,
define an evolution operator as a product of two unitary and self-adjoint 
operators, and compute its spectrum from the transition probabilities
of the random walk. 
The bipartite walk was updated in \cite{MNRS07, MNRS09}
and then reformulated in \cite{Se, HKSS14} as a quantum walk 
on a digraph (without assuming bipartiteness). 
Nowadays, such a generalization is called the (twisted) Szegedy walk
\cite{HKSS13, HKSS14, HS15},
which as a special case includes the Grover walk \cite{Se, Wa}. 
The Szegedy walk on a symmetric digraph $G = (V, D)$ is described 
by the evolution operator $U = SC$,
which is a product of two unitary self-adjoint operators 
$S$ and $C$ on the Hilbert space $\ell^2(D)$
of square summable functions on directed edges $D$. 
Here $S$ and $C$ is called the shift and coin operators. 
Moreover, $C$ can be expressed as $2d^*d-1$,
where $d:\ell^2(D) \to \ell^2(V)$ is coisometry,
{\it i.e.}, $d d^* = 1$,
and is called a boundary operator.
A remarkable feature of the Szegedy walk is
that the spectrum $\sigma(U)$ of $U$ 
can be expressed 
in terms of the discriminant operator $T = dSd^*$ 
and the birth eigenspaces
$\mathcal{B}_\pm = \ker d \cap \ker (S\pm 1)$ as 
\begin{equation}
\label{sgmU} 
\sigma(U) 
	= \varphi^{-1}(\sigma(T)) \cup \{+1\}^{M_+} \cup \{-1\}^{M_-},
\end{equation}
where $\varphi(z) = (z+z^{-1})/2$ and
$M_\pm = \dim \mathcal{B}_\pm$ denotes the cardinality
of the set $\{\pm 1\}$
with the convention $\{\pm 1\}^{M_\pm} = \emptyset$ when $M_\pm=0$. 
This statement is called the spectral mapping theorem 
of quantum walks \cite{Se, HKSS14}
and $\varphi^{-1}(\sigma(T))$ is called the inherited part
\cite{MOS17, HS18}.
In the case of the Grover walk, 
the discriminant operator $T$ is unitarily equivalent to 
the transition probability operator $P$ of the symmetric random walk 
on the graph where the Grover walk itself is defined.  
Hence, the quantized evolution $U$  inherits the spectrum form 
the transition probability operator $P$ of the classical random walk. 
In \cite{HKSS14}, the multiplicities $M_\pm$ were
characterized in terms of graph structure. 

The spectral mapping theorem was extended to 
a more general setting in \cite{HiSeSu, SeSu}. 
Let $d$ be coisometry
from a Hilbert space $\mathcal{H}$ 
to another Hilbert space $\mathcal{K}$ 
and $S$ be a unitary and self-adjoint operator on 
$\mathcal{H}$. 
Then $U := S(2d^*d-1)$ and $T := d Sd^*$ satisfy \eqref{sgmU}. 
The spectral mapping theorem of this form can be applied 
to the spectral analysis of various types of quantum walks. 
Actually, in a previous paper \cite{FFS17}, 
the authors of the current paper 
used it for analyzing 
a $d$-dimensional split-step quantum walk,
which was a unified model 
including Kitagawa's split-step quantum walks \cite{Ki} and 
$d$-dimensional quantum walks \cite{MBSS02,TFMK03,IKK04}. 
In particular, 
the authors performed the spectral analysis of
the inherited part from the discriminant operator $T$
and provided a criterion for $T$ and hence $U$ to have eigenvalues. 

In this paper, 
we perform the analysis of the birth eigenspaces $\mathcal{B}_\pm$
of the one-dimensional split-step quantum walk. 
We provide a criterion for when $\mathcal{B}_\pm$ is nontrivial.  
Moreover, we prove that the norm of vectors in $\mathcal{B}_\pm$ 
(if exists) decay exponentially at spatial infinity
and show the robustness of $\mathcal{B}_\pm$ against perturbations.  
Here we note that the criterion for the nontriviality of $\mathcal{B}_\pm$
is given in terms of
the asymptotic behavior of local coins $C(x)$ as $x$ tends to $\pm \infty$. 
This adapts to  two phase quantum walks \cite{Enetal15,Enetal16}
and anisotropic quantum walks \cite{RiSuTi17a, RiSuTi17b}
and leads us to define a topological index
such as those introduced in \cite{Ki, AsOb13}. 
In a forthcoming paper \cite{Su18},
the third author studies 
such an index in terms of supersymmetric quantum mechanics. 
%

The spectral analysis of the quantum walk is of particular interest,   
because the asymptotic behavior is governed by 
the spectral properties of the evolution operator $U$. 
The presence of an eigenvalue ensures
that localization occurs if and only if the initial state has an overlap 
with its eigenvector \cite{KoLuSe13, SeSu}. 
Hence, if $\mathcal{B}_\pm$ is nontrivial, 
the localization can occur. 
A weak limit theorem originated from Konno \cite{Ko02, Ko05} 
(see also \cite{GrJaSc04})
says that at large time $t$,
the position $X_t$ of the quantum walker
behaves like $X_t \sim t V$,
where $V$ is interpreted as the asymptotic velocity.  
As put into evidence in \cite{Su15, RiSuTi17b}, 
if $U$ is asymptotically homogeneous, 
then 
the distribution $\mu_V$ of $V$ is given by
\[ \mu_V(dv) = \|\Pi_{\rm p}(U) \Psi_0\|^2 \delta_0(dv)
	+\|E_{\hat v}(dv) \Pi_{\rm ac}(U) \Psi_0\|^2. \]
Here $\Pi_{\rm p}(U)$ and  $\Pi_{\rm ac}(U)$ are
orthogonal projections onto the eigenspaces and the subspace of 
absolute continuity for $U$, 
$E_{\hat v}$ is the spectral measure of the velocity operator $\hat v$,
and $\Psi_0$ is the initial state. 
This statement is based on the fact that 
$U$ has no singular continuous spectrum
\cite{AsBoJo15} (see also \cite{RiSuTi17a}). 
A weak limit theorem 
for the one-dimensional split-step quantum walk will be reported 
in a subsequent paper \cite{FFS18b}.  

This paper is organized as follows. 
In Section 2, we define 
a shift operator $S$ and a coin operator $C$
so that both are unitary and self-adjoint 
on $\ell^2(\mathbb{Z};\mathbb{C}^2)$. 
The coin operator $C$ is also assumed to be 
the multiplication operator by 
unitary and self-adjoint matrices $C(x) \in M(2;\mathbb{C})$
($x \in \mathbb{Z}$). 
The evolution operator of the split-step quantum walk 
is defined as $U =SC$ and 
the state of a walker at time $t$ is given by 
$\Psi_t = U^t \Psi_0$, 
where $\Psi_0$ is the initial state of the walker. 
Then the state evolution is 
governed by
\begin{equation*}
\Psi_{t+1}(x) = P(x+1)\Psi_t(x+1) + Q(x-1)\Psi_t(x-1) + R(x)\Psi_t(x), 
	\quad x \in \mathbb{Z}, \ t =0, 1, 2, \ldots, 
\end{equation*}
where $P(x)$, $Q(x)$, and  $R(x) \in M(2; \mathbb{C})$
are determined by $S$ and $C$.  
In Examples \ref{Atype} and \ref{ex:SS},
we see that $U$ becomes the standard one-dimensional quantum walk
and Kitagawa's split-step quantum walk \cite{Ki}
as special cases.

In Section 3, we see that 
the spectral mapping theorem \eqref{sgmU}
is applicable to the split-step quantum walk
and we give an explicit expression of the discriminant operator $T$
in terms of eigenvectors of $C(x)$ (Lemma \ref{exT}). 
Here we also provide a criterion for when $T$ has no eigenvalues 
around $\pm 1$ (Theorem \ref{thm:eig}).

In Section 4, we introduce positive constants $B_\pm$ and $b_\pm$ 
and prove that: 
if $B_\pm < 1$, then $\dim \mathcal{B}_\pm = 1$;
if $b_\pm > 1$, then $\mathcal{B}_\pm$ is a trivial subspace
(Theorem \ref{thm:3}). 
Here, the constants 
$B_\pm$ and $b_\pm$ are defined in terms of 
the asymptotic behavior of local coins $C(x)$ as $x$ tends to $\pm \infty$. 

In Section 5, 
we prove two characteristic properties
of vectors in $\mathcal{B}_\pm$. 
In Subsection 5.1, 
we show that
if $B_\pm < 1$, then $\Psi \in \mathcal{B}_\pm$ exhibits 
an exponential decay, {\it i.e.},
there exist positive constants 
$c_\pm$, $c_\pm^\prime$, 
$\kappa_\pm$, and $\kappa_\pm^\prime$,
such that 
\begin{equation} 
\label{eq_bounds0}
\kappa_\pm^\prime e^{- c_\pm^\prime |x|} 
\leq \|\Psi(x)\|_{\mathbb{C}^2}^2 \leq \kappa_\pm e^{- c_\pm |x|},
\quad |x| \geq R_\pm
\end{equation}
with some $ R_\pm$ sufficiently large.
Let $X_t$ be the random variable denoting the position of 
the quantum walker at time $t$. Then 
the probability distribution of $X_t$ is given by
\begin{equation}
\label{prb0} 
P(X_t = x) = \|\Psi_t(x)\|_{\mathbb{C}^2}^2,
	\quad x \in \mathbb{Z}, \ t = 0,1,2, \ldots,
\end{equation}
where $\Psi_0$ is the initial state. 
Combining \eqref{eq_bounds0}
and \eqref{prb0} 
yields the fact that $P(X_t = x)$ decays exponentially 
for the initial state $\Psi_0 \in \mathcal{B}_\pm$. 
In Subsection 5.2, 
we show that $\mathcal{B}_\pm$ is robust 
against local perturbations of $C(x)$. 
To this end, 
we consider two local coins $C(x)$ and $C^\prime(x)$ that 
satisfy 
$\lim_{x \to \pm \infty} C^\prime(x) = \lim_{x \to \pm \infty} C(x) = :C_{\pm \infty}$,
{\it i.e.}, the difference between $C^\prime(x)$ and $C(x)$ 
vanish at spatial infinity. 
Hence, we can regard $C$ and $C^\prime$ as  
an unperturbed coin and a perturbed coin. 
We use $\mathcal{B}_\pm(C)$ for $\mathcal{B}_\pm$ 
to make the dependence on $C$ explicit. 
We introduce constants $\beta_\pm$ 
determined only by $C_\pm$ and prove that
if $\beta_\pm < 1$, then $\dim \mathcal{B}_\pm(C) = \dim \mathcal{B}_\pm(C^\prime) = 1$
(Theorem \ref{thm:5.2}). 
This implies that $\mathcal{B}_\pm$ are robust against
perturbations that vanish at spatial infinity.

In Section 6,
we give two examples.  
Tne first one is an anisotropic quantum walk.
The second one is Kitagawa's split-step quantum walk. 
In these examples,
we see that the following three cases are possible:
(i) $\dim \mathcal{B}_+ = \dim \mathcal{B}_- = 1$;
(ii) $ \mathcal{B}_+ = \mathcal{B}_- = \{0\}$;
(iii) $\dim \mathcal{B}_\pm = 1$ and $\mathcal{B}_\mp = \{0\}$.

\section{Definition of the model}
Let 
\[ \mathcal{H} := \ell^2(\mathbb{Z}; \mathbb{C}^2) 
	= \{ \Psi :\Z \to \C^2 \mid 
	\sum_{x \in \mathbb{Z}} \|\Psi(x)\|_{\mathbb{C}^2}^2 < \infty \}
\] 
be the Hilbert space of states and define a shift operator $S$
and a coin operator $C$ on $\mathcal{H}$ as follows. 
For a vector
$\Psi = \begin{pmatrix} \Psi_1 \\ \Psi_2 \end{pmatrix} \in \mathcal{H}$ and $x \in \mathbb{Z}$,
\begin{equation}
\label{def_shift} 
(S\Psi)(x) = \begin{pmatrix} 
	p \Psi_1(x)+ q \Psi_2(x+1) \\ 
	\bar{q} \Psi_1(x-1)- p\Psi_2(x) \end{pmatrix},
\end{equation}
where
$(p,q) \in \R \times \C$ satisfies $p^2 + |q|^2 = 1$. 
Then, $S$ is unitary and self-adjoint. 
Let $\{C(x)\}_{x \in \mathbb{Z}} \subset U(2)$ be a family of unitary
and self-adjoint matrices such that
\begin{equation}
\label{def_coin} 
C(x) 
	= \begin{pmatrix} a(x) & b(x) \\
		\overline{b(x)} & -a(x) \end{pmatrix},
\end{equation}
where
$a(x) \in \mathbb{R}$ and $a(x)^2 + |b(x)|^2 =1$. 
Since, by \eqref{def_coin}, ${\rm tr}\, C(x) = 0$ 
and ${\rm det}\, C(x) = -1$,
we have $\ker(C(x) \pm 1) = 1$.
For $\Psi \in \H$, $C\Psi$ is given by
\[ (C \Psi)(x) = C(x) \Psi(x), \quad x \in \mathbb{Z}. \]
Then, $C(x)$ is unitary and self-adjoint and so is $C$. 
We now define an evolution operator as 
\[ U = SC. \] 

Let $\Psi_0 \in \H$ ($\|\Psi_0\|=1$) 
be the initial state of a quantum walker.
We define the state of the walker at time $t \in \N$ as
$\Psi_t = U^t \Psi_0$ and
we obtain the state evolution
\begin{equation}
\label{lazy} 
\Psi_{t+1}(x) = P(x+1)\Psi_t(x+1) + Q(x-1)\Psi_t(x-1) + R(x)\Psi_t(x), 
	\quad x \in \mathbb{Z}, 
\end{equation}
where
\begin{gather*}
P(x) = q \begin{pmatrix}\overline{b(x)} & - a(x) \\ 0 & 0 \end{pmatrix},
\quad  
Q(x) = \bar q \begin{pmatrix} 0 & 0 \\ a(x) & b(x) \end{pmatrix}, \quad
R(x) = p \begin{pmatrix} 
	a(x) & b(x) \\ - \overline{b(x)} & a(x)
	\end{pmatrix}. 
\end{gather*}
From \eqref{lazy}, 
this walk is interpreted as a lazy quantum walk. 
We emphasize that 
our walk is defined as a two-state quantum walk 
on $\ell^2(\mathbb{Z};\mathbb{C}^2)$,
whereas standard lazy quantum walks \cite{IKS05}, \cite{LMZZ15} 
are defined  as  a three-state quantum walk 
on $\ell^2(\mathbb{Z};\mathbb{C}^3)$.
 
Our evolution $U$ partially covers several examples of
one-dimensional two-sate quantum walks 
as seen below. 
\begin{Example}[Ambainis-type QW]
\label{Atype}
		{\rm
		In the one-dimensional quantum walk defined by 
		Ambainis \cite{Am03}, 
		the shift operator
		$S_{\rm A}$ is defined as
		\[ (S_{\rm A} \Psi)(x) 
			= \begin{pmatrix} \Psi_1(x+1) \\ \Psi_2(x-1) \end{pmatrix},
			\quad x \in \mathbb{Z}, \ \Psi \in \mathcal{H}.  \]
		Let $C(x)$ be of the form \eqref{def_coin}
		and set
		 $\tilde C(x) = \sigma_1 C(x)$,
		where
		  $\sigma_1 
			= \begin{pmatrix} 0 & 1 \\ 1 & 0 \end{pmatrix}$.
		Define an evolution operator $U_{\rm A}$ as 
		$U_{\rm A} = S_{\rm A} \tilde C$.
			%
		Then 
		\[ (S_{\rm A}\sigma_1 \Psi)(x) 
			= \begin{pmatrix} \Psi_2(x+1) \\ \Psi_1(x-1) \end{pmatrix},
			\quad \Psi \in \mathcal{H}. \]
		Let $S$ satisfy \eqref{def_shift} with $p=0$ and $q=1$. 
		Then $U$ becomes $U_{\rm A}$. Indeed, 
		$S = S_{\rm A} \sigma_1$ and 
		\[ U = SC = (S_{\rm A} \sigma_1) (\sigma_1 \tilde C) = 
		U_{\rm A}. \]
	We emphasize that the evolution $U_{\rm A}$ is unitarily 
	equivalent to standard quantum walks (see \cite{Oh} for more information). 
		}
	\end{Example}
\begin{Example}[Split-step QW]
	\label{ex:SS}
	{\rm
	Let $S_+$ and $S_-$ be shift operators defined as
	\[ (S_+ \Psi)(x) 
			= \begin{pmatrix} \Psi_1(x-1) \\ \Psi_2(x) \end{pmatrix},
			\quad (S_- \Psi)(x) 
			= \begin{pmatrix} \Psi_1(x) \\ \Psi_2(x+1) \end{pmatrix},
			\quad x \in \mathbb{Z}, \ \Psi \in \mathcal{H}. \]
	The evolution $U_{\rm ss}(\theta_1, \theta_2)$ 
	of the split-step quantum walk introduced by Kitagawa 
	{\it et al} \cite{Ki} is defined as
	\[ U_{\rm ss}(\theta_1, \theta_2)
			= S_- R(\theta_2) S_+ R(\theta_1), \]
			where $\theta_1, \theta_2 \in [0,2\pi)$ and
	\[ R(\theta) 
			= \begin{pmatrix} \cos(\theta/2) & -\sin(\theta/2) \\
			\sin(\theta/2) & \cos(\theta/2) \end{pmatrix}. \]
	By direct calculation, 
	\[ (\sigma_1 S_-  R(\theta)S_+\Psi)(x)
= \begin{pmatrix} 
			p_\theta \Psi_1(x) +  q_\theta \Psi_2(x+1) \\
			q_\theta \Psi_1(x-1) - p_\theta \Psi_2(x) 
			\end{pmatrix}
	\quad 
	\mbox{with $p_\theta = \sin(\theta/2)$ 
		and $q_\theta = \cos(\theta/2)$}\]
			and
	\[ R(\theta) \sigma_1 
		= \begin{pmatrix}  a_\theta & b_\theta  \\
		b_\theta & -a_\theta  \end{pmatrix}
	\quad 
	\mbox{with $a_\theta = - \sin(\theta/2)$ and $b_\theta = \cos(\theta/2)$}. 
			\]
When $p =  p_{\theta_2}
$ and $q = q_{\theta_2}
$,
			$S = \sigma_1 S_-  R(\theta_2)S_+$.
Let $C(x) = R(\theta_1) \sigma_1$. 
			Then, 
			\[ U = SC = \sigma_ 1 U_{\rm ss}(\theta_1,\theta_2) \sigma_1. \]
			Hence, $U$ and $U_{\rm ss}(\theta_1,\theta_2)$ are 
			unitarily equivalent. 
		}
	\end{Example}

We call the quantum walk with the evolution $U$
a {\it split-step quantum walk},
because it is a generalization of Kitagawa's split-step quantum walk
in the sense of Example \ref{ex:SS}. 
Throughout this paper, 
we consider the shift oprator $S$ and coin operator $C$ defined by
\eqref{def_shift} and \eqref{def_coin}
unless otherwise stated. 

\section{Spectral mapping theorem}
In this section, we apply spectral mapping techniques 
obtained in \cite{HiSeSu, SeSu} 
for the product of two unitary self-adjoint operators.  
Since the shift $S$ and coin $C$ are unitary and self-adjoint,
these techniques can be applied to the evolution $U$ 
of the split step quantum walk. 
To applying these techniques,
we need to define a coisometry operator 
\[ 
d:\H \rightarrow \ell ^2(\Z) = :\mathcal{K}
\] 
so that $C=2d^*d-1$.
Since $\dim \ker \( C(x)-1\)=1$, 
we can choose a unique normalized vector 
$\chi (x)=\begin{pmatrix}
	\chi _1(x)\\
	\chi _2(x)
	\end{pmatrix}\in \ker \( C(x)-1 \)$
up to a constant factor. 
We now define an operator $d:\mathcal{H} \to \mathcal{K}$ as 
\begin{equation}
\label{eq:3.1}
\(d\Psi \)(x)
=\i<\chi (x),\Psi (x)>_{\C^2},\quad x\in \Z
\quad \mbox{for $\Psi \in \H$}.
\end{equation}
We use ${\rm id}_\mathcal{K}$ to denote
the identity on $\mathcal{K}$. 
\begin{Lemma}
Let $d$ be as above. 
\begin{itemize}
\item[(1)] $d$ is bounded and its adjoint 
$d^*:\mathcal{K} \to \mathcal{H}$ is given by 
\[ d^*\psi = \psi \chi \quad \mbox{for $\psi \in \mathcal{K}$}.  \] 
\item[(2)] $d$ is coisometry, {\it i.e.}, 
	$d d^* = {\rm id}_\mathcal{K}$.
\item[(3)] $C = 2d^*d-1$. 
\end{itemize}
\end{Lemma}
\begin{proof}
Because $\chi(x)$ is a normalized vector in $\mathbb{C}^2$,
\begin{align*}
\|d \Psi \|_\mathcal{K}^2
 = \sum_{x \in \mathbb{Z}}
 	\left|\langle \chi (x), \Psi(x) \rangle_{\mathbb{C}^2}\right|^2
\leq \|\Psi\|_\mathcal{H}^2
\quad \mbox{for $\Psi \in \mathcal{H}$}. 
\end{align*}
Hence, $d$ is bounded. 
For 
$\psi \in \mathcal{K}$,
\begin{align*}
\langle \psi, d\Psi \rangle_\mathcal{K}
= \sum_{x \in \mathbb{Z}} 
		\bar\psi(x) \langle \chi(x), \Psi(x) \rangle_{\mathbb{C}^2}
= \langle \psi \chi, \Psi \rangle_\mathcal{H},
\end{align*}
which completes (1). 

By direct calculation,
\[ (dd^* \psi)(x)
	= \langle \chi(x), (d^*\psi)(x) \rangle_{\mathbb{C}^2}
	= \psi(x), \quad x \in \mathbb{Z}, 
\]
which implies that $d$ is coisometry. 
Hence, (2) is proved. 

Because 
$d^*d 
= \sum_{x \in \mathbb{Z}} |\chi(x)\rangle \langle \chi(x)|$
and
$
	C(x)=2|\chi (x)\rangle \langle \chi (x)|-1
$, 
we obtain (3).  
\end{proof}

With the terminology of \cite{Sz, HS15}, 
we call $T = d S d^*$ the {\it discriminant} of $U$ 
and 
\begin{equation} 
\label{birth}
\mathcal{B}_\pm 
	= \ker d \cap \ker (S \pm 1) 
\end{equation}
the {\it birth eigenspaces}. 
We set
$M_\pm  = \dim \mathcal{B}_\pm$.  
Let $\varphi$ be
the Joukowsky transformation:
\[ \varphi(z) = \frac{z + z^{-1}}{2}, \]
which maps 
$S^{1}:= \{ z \in \mathbb{C} \mid |z|=1 \}$ 
onto $[-1,1]$. 
We use $\sigma(A)$ to denote the spectrum of an operator $A$
and $\sigma_{\rm p}(A)$ the point spectrum of $A$. 
The following proposition is a direct consequence of \cite{HiSeSu}.
\begin{Proposition}[Spectral mapping theorem \cite{HiSeSu}] 
\label{prop:smt}
Let $U$, $T$, $d$,  $M_\pm$,  and $\varphi$ as above. 
\begin{itemize}
\item[(1)] $T$ is bounded and self-adjoint on $\mathcal{K}$ with $\|T\|  \leq 1$. 
\item[(2)] $\sigma(U) 
	= \varphi^{-1}(\sigma(T)) \cup \{1\}^{M_+} \cup \{-1\}^{M_-}$.
\item[(3)] $\sigma_{\rm p}(U) 
	= \varphi^{-1}(\sigma_{\rm p}(T)) \cup \{1\}^{M_+} \cup \{-1\}^{M_-}$. 
\item[(4)] $\dim \ker (U\mp 1) 
= M_\pm + m_\pm$, where $m_\pm = \dim \ker (T\mp 1)$. 
\end{itemize}
\end{Proposition}

We use $L$ to denote the left shift operator on $\mathcal{K}$:
\[ (L \psi)(x) = \psi(x+1), \quad x \in \mathbb{Z}
	\quad \mbox{for $\psi \in \mathcal{K}$}. \]
\begin{Lemma}
\label{exT}
Let $V$ denote the multiplication operator by
\[ V(x) = p\left(|\chi_1(x)|^2 - |\chi_2(x)|^2 \right)  \]
and $D = q \bar \chi_1 L \chi_2$.  
Then 
\[ T = D + D^* + V.  \]
\end{Lemma}
\begin{proof}
The proof proceeds along the same lines as the proof of
\cite{FFS17}[Lemma 3.2]. 
Identifying  $\mathcal{H}$ with  $\mathcal{K} \oplus \mathcal{K}$,
we observe that 
$S = \begin{pmatrix}
	p & q L \\
	\bar q L^* & - p
	\end{pmatrix}$.  
Hence, for $\psi \in \mathcal{K}$, 
$S d^* \psi = S \psi \chi 
= \begin{pmatrix}
	(p \chi_1  + q L \chi_2 )\psi \\
	(\bar q L^* \chi_1 - p\chi_2) \psi
	\end{pmatrix}$
and 
\begin{align*}
T\psi 
& = \langle \chi(\cdot), (S d^*\psi)(\cdot) \rangle_{\mathbb{C}^2} \\
& = \bar \chi_1 (p \chi_1  + q L \chi_2) \psi
		+ \bar \chi_2 (\bar q L^* \chi_1 - p\chi_2) \psi \\
& =  q \bar \chi_1 L \chi_2 \psi 
	+ \bar q \bar \chi_2 L^* \chi_1 \psi 
	+  p(|\chi_1|^2 - |\chi_2|^2 ) \psi.  
\end{align*}
This completes the proof. 
\end{proof}
	
In what follows, 
we provide a criterion for when $T$ has no eigenvalues 
$\pm E $ with $E > |V|_\infty := \sup_{x \in \mathbb{Z}} |V(x)|$. 
Because, for such an $E$, $E \mp V(x) > 0$ ($x \in \mathbb{Z}$), 
we can define an operator $K_E$ as
\[ K_E^\pm 
	= \frac{1}{\sqrt{E \mp V}} (D + D^*) 
		 \frac{1}{\sqrt{E \mp V}}.  \]
\begin{Lemma}
\label{lem_egnT}
Let $E$ be as above. 
The following are equivalent. 
\begin{itemize}
\item[(i)] $\pm E \in \sigma_{\rm p}(T)$.
\item[(ii)] $ \pm 1 \in \sigma_{\rm p}(K_E^\pm)$. 
\end{itemize} 
In this case, $\dim \ker(T \mp E) = \dim \ker (K_E \mp 1)$. 
\end{Lemma}
\begin{proof}
The assertion of the lemma follows from
\[ T \mp E = \sqrt{E \mp V} (K_E \mp 1) \sqrt{E \mp V},
	\quad E > |V|_\infty.  \]
\end{proof}
\begin{Theorem}
\label{thm:eig}
If $E > |q| + |V|_\infty$,
then $\pm E \not\in \sigma_{\rm p}(T)$. 
\end{Theorem}
\begin{proof}
By Lemma \ref{lem_egnT}, it suffices to prove that
$\|K^\pm_E\| < 1$. To this end, we consider the rage of
$\langle \psi, K^\pm_E \psi \rangle$ for $\psi \in \mathcal{K}$. 
Let $\psi_i = \frac{\chi_i}{\sqrt{E \mp V}} \psi$ ($i =1,2$). 
Then
\begin{align*}
|\langle \psi, K^\pm_E \psi \rangle|
& = 2 |q| \|\psi_1\| \|\psi_2 \| \\
& \leq |q| \left( \|\psi_1\|^2 +\|\psi_2\|^2 \right). 
\end{align*}
Because $|\chi_1(x)|^2 + |\chi_2(x)|^2 = 1$,
\[ \|\psi_1\|^2 +\|\psi_2\|^2 
	= \sum_{x \in \mathbb{Z}} \frac{|\psi(x)|^2}{E \mp V(x)}.
\]
Hence, $\|K^\pm_E\| \leq |q|/(E - |V|_\infty)$
and the proof of the lemma is complete. 
\end{proof}

\section{Nontriviality of the birth eigenspaces}
In this section, we address
the problem when the birth eigenspaces
$\mathcal{B}_\pm$ defined in \eqref{birth} becomes nontrivial. 
To this end, we characterize $\mathcal{B}_\pm$. 

In the case of $|p| = 1$, 
$S$ becomes a constant matrix
and $U$ becomes a multiplication operator.  
In this case, 
the quantum walker never moves and 
hence the quantum walk becomes trivial
(see also \eqref{lazy} with $q=0$). 
To avoid this trivial case,
we suppose the following. 
\begin{Assumption}\label{ass1}
	$|p|\not =1$.
\end{Assumption}
\begin{Lemma}\label{thm:1-2}
Assume Hypothesis 1. Then 
\[
				\ker (S\pm 1)=\left\{ \begin{pmatrix}
				-\frac{q}{p\pm 1}L\psi \\
				\psi 	\end{pmatrix}
			~\Bigg|~ \psi \in \mathcal{K}
				\right\}.
\]
\end{Lemma} 
\begin{proof}
Let 
$\Psi 
=\begin{pmatrix}
	\Psi _1\\
	\Psi _2
\end{pmatrix} \in \mathcal{H}$.
Because 
$
S\pm 1=
		\begin{pmatrix}
			p\pm 1&qL\\
			\bar{q}L^*&-p\pm 1
		\end{pmatrix}
$, we observe that $\Psi \in \ker( S \pm 1)$ if and only if
$\Psi_1$ and $\Psi_2$ belong to $\mathcal{K}$ and satisfy

\begin{equation}
\label{eq:lem4.1:1}
	\begin{cases}
		(p\pm 1)\Psi _1(x)+q\(L\Psi _2\)(x)=0,\\
		\bar{q}\(L^*\Psi _1\)(x)+(-p\pm 1)\Psi _2(x)=0
	\end{cases}
	\quad \mbox{for all $x \in \mathbb{Z}$.}
\end{equation}
By Hypothesis \ref{ass1}, 
\eqref{eq:lem4.1:1} is equivalent to
\[
	\begin{cases}
		\Psi _1(x)=-\frac{q}{p\pm 1}\(L\Psi _2\)(x),\\
		\( -\frac{|q|^2}{p\pm 1}-p\pm 1 \)\Psi _2(x)=0
	\end{cases}
	\quad \mbox{for all $x \in \mathbb{Z}$.}
\]
Because $p^2+|q|^2=1$ implies that
\[	
-\frac{|q|^2}{p\pm 1}-p\pm 1 =0,
\]
we obtain the desired result. 
\end{proof}
Combining Lemma \ref{thm:1-2} with \eqref{eq:3.1}
yields the following.
\begin{equation}\label{eq:2-2}
	\mathcal{B}_\pm
	=\Bigl\{ \Psi =\begin{pmatrix}
	-\frac{q}{p\pm 1}L\psi \\
	\psi 
\end{pmatrix}
~\Big|~ \psi \in \mathcal{K}, \
- q \overline{\chi _1}L\psi +(p\pm1) \overline{\chi _2}\psi=0
\Bigr\} .
\end{equation}
In order to provide a criterion for $\mathcal{B}_\pm$ 
to be nontrivial,
we suppose the following.  
\begin{Assumption}\label{ass2}
$\chi _1(x)\chi _2(x)\not =0$ for all $x \in \mathbb{Z}$.
\end{Assumption}
We define four constants 
$B_\pm$ and $b_\pm$
as 
\[ B_\pm = \max \{B_\pm(-\infty),  B_\pm(+\infty)\},
	\quad b_\pm = \min \{b_\pm(-\infty),  b_\pm(+\infty)\},  \]
where 
\begin{align*}
& 
B_\pm(- \infty)
= \limsup _{x\rightarrow -\infty }
\left| \frac{q{\chi _1}(x)}{\( p \pm 1 \) {\chi _2}(x)} \right|^2,
\quad
B_\pm( + \infty)
= \limsup _{x\rightarrow + \infty }
\left| \frac{\( p\pm1 \){\chi _2}(x)}{q{\chi _1}(x)}\right|^2,
\\[2mm]
& 
b_\pm(- \infty)
= \liminf_{x\rightarrow -\infty }
\left| \frac{q{\chi _1}(x)}{\( p \pm 1 \) {\chi _2}(x)} \right|^2,
\quad
b_\pm( + \infty)
= \liminf _{x\rightarrow + \infty }
\left| \frac{\( p\pm1 \){\chi _2}(x)}{q{\chi _1}(x)}\right|^2.
\end{align*}

We are now in a position to state our main result.
\begin{Theorem}\label{thm:3}
Assume Hypotheses \ref{ass1} and \ref{ass2}. 
\begin{itemize}
\item[(1)] If $B_\pm< 1$, then 
$\dim \mathcal{B}_\pm = 1$. 
\item[(2)] 
If $b_\pm> 1$, then $\mathcal{B}_\pm= \{0\}$.
\end{itemize}
\end{Theorem}
In order to prove Theorem \ref{thm:3},
we use the following lemma. 
\begin{Lemma}
\label{lem:4.3}
Assume Hypotheses \ref{ass1} and \ref{ass2}. 
Let $\psi:\mathbb{Z} \to \mathbb{C}$ 
be a nonzero solution to 
\begin{equation}
\label{eq:3-3} 
L\psi =\frac{\( p\pm1\) \overline{\chi_2}}{q\overline{\chi _1}}\psi.
\end{equation}
Then
\begin{align} 
\label{eq:B}
& B_\pm(-\infty) 
= \limsup_{x \to -\infty} \Bigl| \frac{\psi (x-1)}{\psi (x)} \Bigr|^2,
\quad
 B_\pm(+\infty) 
 = \limsup _{x\rightarrow +\infty}\Bigl| \frac{\psi (x+1)}{\psi (x)} \Bigr|^2, \\
\label{eq:b}
& b_\pm(-\infty) 
= \liminf_{x \to -\infty} \Bigl| \frac{\psi (x-1)}{\psi (x)} \Bigr|^2,
\quad
b_\pm(+\infty) 
 = \liminf _{x\rightarrow +\infty}\Bigl| \frac{\psi (x+1)}{\psi (x)} \Bigr|^2.
\end{align} 
\end{Lemma}
\begin{proof}
Because Hypotheses \ref{ass1} and \ref{ass2} 
imply that $p \pm 1$, $q$,
$\chi_1(x)$, and  $\chi_2(x)$ are not zero,
\eqref{eq:3-3} is equivalent to
\begin{equation} 
\label{eq:3-4}
\begin{cases}
		\displaystyle	\psi (x+1)=\frac{\( p\pm1 \)\overline{\chi _2}(x)}{q\overline{\chi _1}(x)}\psi (x),\ \ x\geq 0,\\
		\vspace{1pt}\\
		\displaystyle \psi (x-1)=\frac{q\overline{\chi _1}(x)}{\( p\pm1 \) \overline{\chi _2}(x)}\psi (x),\ \ x\leq 0.
	\end{cases}
\end{equation}
Since $\psi \not\equiv 0$,
\eqref{eq:3-4} implies that $\psi(x) \not=0$ for all $x \in \mathbb{Z}$. 
Hence,
\begin{align*}
&
\left|\frac{q{\chi _1}(x)}{\( p\pm1 \) {\chi _2}(x)} \right|^2
= \Bigl| \frac{\psi (x-1)}{\psi (x)} \Bigr|^2,
\quad x \leq 0, \\
&
\left|\frac{\( p\pm1 \){\chi _2}(x)}{q{\chi _1}(x)}
 \right|^2 
 = \Bigl| \frac{\psi (x+1)}{\psi (x)} \Bigr|^2,
 \quad x > 0.
 \end{align*}
Taking the limits of both sides, we obtain the desired result. 
\end{proof}
\begin{proof}[Proof of Theorem \ref{thm:3}]
By \eqref{eq:2-2}, 
$\Psi \in \mathcal{B}_\pm$
if and only if there exists a vector $\psi \in \mathcal{K}$ such that
$\Psi =\begin{pmatrix}
-\frac{q}{p\mp 1}L\psi \\
\psi 	\end{pmatrix}$ 
and 
$\psi$ satisfies \eqref{eq:3-3}. 
Now we suppose that $B_\pm < 1$. 
We define a function $\psi_0:\mathbb{Z} \to \mathbb{C}$ 
inductively as follows. 
Let $\psi_0(0) = 1$ and define $\psi_0(x)$ ($x\not=0$) 
by \eqref{eq:3-4}. 
Then, from the above argument, $\psi_0$ satisfies \eqref{eq:3-3}.  
By definition, $\psi_0 \not\equiv 0$. 
Hence, we have $\psi_0 \in \mathcal{K}$
by combining Lemma \ref{lem:4.3} with the ratio test.  
Thus, defining 
$\Psi_0 
=\begin{pmatrix}
-\frac{q}{p\mp 1}L\psi_0 \\
\psi_0 
\end{pmatrix}$,
we observe that $\Psi_0$ is nonzero 
and belongs to $\mathcal{B}_\pm$.  
Hence, $\mathcal{B}_\pm$ is nontrivial. 
If there is another nonzero vector  
$\Psi =\begin{pmatrix}
-\frac{q}{p\mp 1}L\psi \\
\psi 	\end{pmatrix} \in \mathcal{B}_\pm$,
then $\psi$ also satisfies \eqref{eq:3-4}. 
Taking a constant  $\alpha = \psi(0)/\psi_0(0)$,
we observe from \eqref{eq:3-4} that $\psi = \alpha \psi_0$.
Hence, $\dim \mathcal{B}_\pm = 1$. 
Thus, (1) is proved. 

We next suppose that $b_\pm > 1$ and 
$\Psi =\begin{pmatrix}
-\frac{q}{p\mp 1}L\psi \\
\psi 	\end{pmatrix} \in \mathcal{B}_\pm$ is nonzero. 
Similarly to the above,
the ratio test implies that  $\psi \not \in \mathcal{K}$. 
This contradicts $\Psi \in \mathcal{B}_\pm$.  
Hence, $\mathcal{B}_\pm = \{0\}$.  
Thus, (2) is proved. 
\end{proof}
%
%
%
\section{Properties of vectors in the birth eigenspaces}
Throughout this section, 
we suppose that Hypotheses \ref{ass1} and \ref{ass2} are satisfied. 
Summarizing the arguments in Sec. 4,
we observe that
\begin{equation}
\label{eq:5.1}
\mathcal{B}_\pm
= \left\{ 
\Psi = \begin{pmatrix} - \frac{q}{p \pm 1} L\psi \\ \psi \end{pmatrix}
~ \Bigg|~ 
\mbox{$\psi \in \mathcal{K}$ satisfies \eqref{eq:3-4}}
\right\}.
\end{equation}
\subsection{Exponential decay}
We prove that the birth eigenvector $\Psi \in \mathcal{B}_\pm$ 
decays exponentially at spatial infinity.  
\begin{Theorem}
Suppose that $B_\pm < 1$ and $\Psi \in \mathcal{B}_\pm$. 
Then, there exist
positive constants $c_\pm$, $c_\pm^\prime$,
$\kappa_\pm$, $\kappa_\pm^\prime$, and $R_\pm > 0$ 
such that
\begin{equation} 
\label{eq_bounds}
\kappa_\pm^\prime e^{- c_\pm^\prime |x|} 
\leq \|\Psi(x)\|_{\mathbb{C}^2}^2 \leq \kappa_\pm e^{- c_\pm |x|},
	\quad |x| \geq R_\pm. 
\end{equation}
\end{Theorem}
\begin{proof}
Let $\Psi \in \mathcal{B}_\pm$. 
From \eqref{eq:5.1}, 
there exists a $\psi \in \mathcal{K}$ such that
$\Psi  = \begin{pmatrix} - \frac{q}{p \pm 1} L\psi \\ \psi \end{pmatrix} \in \mathcal{B}_\pm$ and 
$\psi$ satisfy \eqref{eq:3-4}. 

We first prove the right-hand side of \eqref{eq_bounds}. 
Because $\|\Psi(x)\|_{\mathbb{C}^2}^2 
= (|q|^2/(p \pm 1)^2) |\psi(x+1)|^2 + |\psi(x)|^2$,
it suffices to prove that
\[ |\psi(x)|^2 \leq 
\kappa_\pm e^{- c_\pm |x|},
	\quad |x| \geq R_\pm \]
with some $\kappa_\pm$, $c_\pm$, and $R_\pm > 0$. 
Let $\epsilon$ satisfy $0 < \epsilon < 1- B_\pm(+\infty)$. 
By the definition of $B_\pm(+\infty)$,
there exists $x_0 \in \mathbb{N}$ such that 
if $x \geq x_0$,
\[
0 \leq \sup_{y \geq x} 
\left|\frac{\( p\pm1 \)\overline{\chi _2}(y)}
{q\overline{\chi _1}(y)}\right|^2 -  B_\pm(+\infty)  
\leq  \epsilon.
\]
Suppose that $x \geq x_0$. By \eqref{eq:3-4}, 
\begin{align*}
|\psi(x)| 
& = 
\left|\frac{\( p\pm1 \)\overline{\chi _2}(x-1)}
{q\overline{\chi _1}(x-1)}\right||\psi(x-1)| \\
& =
\left|\frac{\( p\pm1 \)\overline{\chi _2}(x-1)}
{q\overline{\chi _1}(x-1)}\right|
\cdot
\left|\frac{\( p\pm1 \)\overline{\chi _2}(x-2)}
{q\overline{\chi _1}(x-2)}\right|
\cdot \cdots \cdot
\left|\frac{\( p\pm1 \)\overline{\chi _2}(x_0)}
{q\overline{\chi _1}(x_0)}\right| |\psi(x_0)| \\
& \leq
(B_\pm(+\infty) + \epsilon)^{(x-x_0)/2} |\psi(x_0)|
\end{align*}
Since $B_\pm(+\infty) + \epsilon < 1$,
\[ c_\pm(+\infty) 
	:= - 
		\log (B_\pm(+\infty) + \epsilon) >0. \]
Hence, if $x \geq x_0$,
\begin{align*}
|\psi(x)|^2 \leq \kappa_\pm(+\infty)
	e^{- c_\pm(+\infty) x}
\end{align*}
with $ \kappa_\pm(+\infty) := |\psi(x_0)|^2 e^{ c_\pm(+\infty)x_0 }$.
Similarly, taking $0 < \epsilon  < 1- B_\pm(-\infty)$ and
$c_\pm(-\infty) = - \log(B_\pm(-\infty)+\epsilon)$,
we can prove that there exists $x_1 \in \mathbb{N}$  such that
if $x \leq - x_1$,
\begin{align*}
|\psi(x)|^2 \leq \kappa_\pm(-\infty)
	e^{- c_\pm(-\infty) |x|}
\end{align*}
with $ \kappa_\pm(-\infty) := |\psi(x_1)|^2 e^{ c_\pm(-\infty)|x_1| }$.
Tanking 
$c_\pm = \min\{c_\pm(-\infty), c_\pm(+\infty) \}$, 
$\kappa_\pm = \max\{\kappa_\pm(-\infty),$ \\ 
$\kappa_\pm(+\infty) \}$,
$R_\pm := \max \{x_0, x_1\}$ yields the right-hand side
of \eqref{eq_bounds}.  

Next we prove the left-hand side of \eqref{eq_bounds}. 
Taking $\epsilon$ as $0 < \epsilon < b_\pm(+\infty)$,
we have 
\[ 0 \leq b_\pm(+ \infty) 
- \inf_{y \geq x}
\left|\frac{\( p\pm1 \)\overline{\chi _2}(y)}
{q\overline{\chi _1}(y)}\right|^2
\leq \epsilon
\]
and
\[ |\psi(x)| 
\geq
(b_\pm(+\infty) - \epsilon)^{(x-x_0)/2} |\psi(x_0)|
\]
for $x$ greater than  some $x_0 >0$. 
Because $0 < b_\pm(+\infty) - \epsilon < 1$,
\[ c_\pm^\prime(+\infty) 
	:= - 
		\log (b_\pm(+\infty) - \epsilon) >0, \]
we observe from the same argument as above that
\[ |\psi(x)|^2 \geq \kappa_\pm^\prime(+\infty)
	e^{- c_\pm^\prime(+\infty) x},
	\quad x \geq x_0 \]
with some $ \kappa_\pm^\prime(+\infty) > 0$. 
Thus, the-right hand side of \eqref{eq_bounds} is proved for $x > 0$ sufficiently large. 
For $x \leq 0$ sufficiently small, 
the same argument works. 
Therefore, we complete the proof. 
\end{proof}
\subsection{Robustness against perturbations}
In this subsection, we illustrate 
the robustness of the birth eigenspaces against perturbations. 
For simplicity, 
we focus here on the case in which the limits 
$\lim_{x \to \pm \infty} C(x)$ exist. 
This is a slight generalization of 
the anisotropic quantum walk introduced 
in \cite{RiSuTi17a, RiSuTi17b}, 
where the authors addressed the case of $p=0$. 
We address the case in which 
$p$ can be  nonzero. 

Let $C_{\pm \infty}\in U(2)$ be self-adjoint unitary matrices
with ${\rm det}\, C_{\pm \infty} = -1$ 
and choose normalized vectors 
$\chi_{\pm \infty}
= \begin{pmatrix} \chi_{{\pm \infty},1} \\ \chi_{{\pm \infty}, 2} \end{pmatrix}$ 
such that
\[ C_{\pm \infty} = 2|\chi_{\pm \infty} \rangle \langle \chi_{\pm \infty}|-1. \]
We define two constants $\beta_+$ and $\beta_-$ as
\[ \beta_+ = \max \{\beta_+(-\infty),  \beta_+(+\infty)\},
	\quad \beta_- = \min \{\beta_-(-\infty),  \beta_-(+\infty)\},  \]
where
\begin{align*}
& 
\beta_\pm(- \infty)
= 
\left| \frac{q \chi _{-\infty,1}}{\( p\pm 1 \) \chi _{-\infty,2}} \right|^2,
\quad
\beta_\pm( + \infty)
= 
\left| \frac{\( p\pm 1 \) \chi _{+\infty,2}}{q\chi _{+\infty,1}}\right|^2.
\end{align*}

\begin{Theorem}
\label{thm:5.2}
Let $C(x)$ be defined in \eqref{def_coin} and satisfy in addition
$\lim_{x \to \pm \infty} C(x) = C_{\pm \infty}$.
Then $B_\pm = b_\pm = \beta_\pm $. 
In particular, the following hold. 
\begin{itemize}
\item[(1)] If $\beta_\pm< 1$, then 
$\dim \mathcal{B}_\pm = 1$. 
\item[(2)] 
If $\beta_\pm> 1$, then $\mathcal{B}_\pm = \{0\}$.
\end{itemize}
\end{Theorem}
\begin{proof}
By assumption, $\lim_{\infty \to \pm x}|\chi_j(x)|= |\chi_{\pm\infty,j}|$
and hence $B_\pm = b_\pm = \beta_\pm$.  
Therefore, by Theorem \ref{thm:3}, we obtain the desired result. 
\end{proof}
We write $\mathcal{B}_\pm(C)$
for the birth eigenspaces $\mathcal{B}_\pm$ 
to make the dependence on the coin $C$ explicit. 
The following is a direct consequence of Theorem \ref{thm:5.2} 
and reveals the robustness of the birth eigenspaces
against perturbations.  
See Examples in the next section for more details. 
\begin{Corollary}
\label{cor:5.3}
Let $C(x)$ be as  in Theorem \ref{thm:5.2} 
and $C^\prime(x)$ 
satisfy the same condition as $C(x)$, so that
$\lim_{x \to \pm \infty}  \left(C^\prime(x) - C(x) \right) = 0$.  
\begin{itemize}
\item[(i)] If $\beta_\pm < 1$, then 
$\dim \mathcal{B}_\pm(C^\prime) = \dim \mathcal{B}_\pm(C) = 1$.
\item[(ii)] If $\beta_\pm > 1$, then 
$ \mathcal{B}_\pm(C^\prime) =  \mathcal{B}_\pm(C) = \{0\}$.
\end{itemize}
\end{Corollary}
\begin{proof}
Because the assertions of Theorem \ref{thm:5.2}
depend only on the limit of the coin operator,
and both $C^\prime(x)$ and $C(x)$ have the same limits,
we obtain the desired result.  
\end{proof}
%
%
%
\section{Examples}
In this section, we provide examples. 

\subsection{An anisotropic coin model}

Let $\varepsilon>0$ and define
\[
	C_{+\infty}:=\begin{pmatrix}
		1-2\varepsilon^2&2\varepsilon \sqrt{1-\varepsilon^2}\\
		2\varepsilon \sqrt{1-\varepsilon^2}&2\varepsilon^2-1
	\end{pmatrix},\hspace{10pt}C_{-\infty}:=\begin{pmatrix}
	2\varepsilon^2-1&2\varepsilon \sqrt{1-\varepsilon^2}\\
	2\varepsilon \sqrt{1-\varepsilon^2}&1-2\varepsilon^2
\end{pmatrix}.
\]
We can choose $\chi_{\pm\infty} \in \ker (C_{\pm \infty} -1)$ as follows. 
\[
	\chi _{+\infty} =\begin{pmatrix}
		\sqrt{1-\varepsilon^2}\\
		\varepsilon
	\end{pmatrix},\hspace{10pt}
	\chi _{-\infty} =\begin{pmatrix}
	\varepsilon\\
	\sqrt{1-\varepsilon^2}
\end{pmatrix}
\]
are eigenvectors of $C_{\pm \infty}$ 
corresponding to the eigenvalues $1$.
Let $\{\chi (x) \}\subset \C^2$ 
be a family of normalized vectors and satisfy 
$\chi_1(x)\chi_2(x) \not=0$ and
$\lim_{x \to \pm \infty} \chi(x) = \chi_{\pm \infty}$. 
Then $C(x) :=2|\chi (x)\rangle \langle \chi (x)|-1$ satisfies
\begin{equation}
\label{eq:6.0}
	\lim _{x\rightarrow \pm\infty}C(x)=C_{\pm \infty}.
\end{equation}
By direct calculation,
\begin{equation}
\label{eq:6.1}
\beta_\pm(-\infty)
= g(\epsilon) \frac{1 \mp p}{1 \pm p}  ,
\qquad 
\beta_\pm(+\infty)	
=  g(\epsilon)  \frac{1 \pm p}{1 \mp p},
\end{equation}
where $g(\epsilon) := \epsilon^2/(1-\epsilon^2)$. 
\begin{Theorem}
\label{thm:6.1}
Let $\epsilon_0 \in (0,1)$ be a unique solution to
\begin{equation}
\label{eq:6.2}
g(\epsilon) = \min \left\{ \frac{1 - p}{1 + p},
	\frac{1 + p}{1 - p} \right\}.
\end{equation}
\begin{itemize}
\item[(1)] If $\epsilon < \epsilon_0$, 
then $\dim \mathcal{B}_+ = \dim \mathcal{B}_- = 1$.
\item[(2)] If $\epsilon > \epsilon_0$,
then $\mathcal{B}_+ = \mathcal{B}_- = \{0\}$. 
\end{itemize}
\end{Theorem}
\begin{Remark}
Combining Theorem \ref{thm:6.1} with Proposition \ref{prop:smt},
we observe that
\[ \dim \ker(U \mp 1) \geq \dim \mathcal{B}_\pm = 1
	\quad \mbox{for $\epsilon < \epsilon_0$}. \]
Moreover, the above statement is independent of
the choice of $C(x) = 2|\chi(x)\rangle \langle \chi(x)|-1$ 
that satisfies \eqref{eq:6.0}. 
See Corollary \ref{cor:5.3}.  
In the case of the two phase quantum walk 
with
\[ C(x) = \begin{cases} C_{+\infty}, & x > 0, \\ 
	C_{-\infty}, & x \leq 0 \end{cases} \]
and $\epsilon < 1/\sqrt{2}$, 
Theorem \ref{thm:eig} implies that if $|q| < 2\epsilon^2$,
then $\pm 1 \not \in \sigma_{\rm p}(T)$. 
Hence, 
\[ \dim \ker(U \mp 1) = 1
	\quad \mbox{for $\epsilon < \min \{\epsilon_0, 1/\sqrt{2} \}$}. \]
\end{Remark}
\begin{proof}
Because $g$ is strictly increasing in $(0, 1)$
with 
$\lim_{\epsilon  \downarrow 0}g(\epsilon) = 0$
and $\lim_{\epsilon \uparrow 1}g(\epsilon) = + \infty$,
there is a unique solution $\epsilon_0$ to \eqref{eq:6.2} in $(0,1)$. 
By \eqref{eq:6.1}, 
\[ \beta_+ =  g(\epsilon)  \max\left\{ \frac{1 - p}{1 + p},
	\frac{1 + p}{1 - p} \right\} = \beta_-. 
\]
Hence, $\beta_\pm < 1$ if and only if 
\[ g(\epsilon) <  \min\left\{ \frac{1 - p}{1 + p},
	\frac{1 + p}{1 - p} \right\} = g(\epsilon_0). \]
Therefore, Theorem \ref{thm:5.2} provides the desired results.
\end{proof}
\subsection{Kitagawa's split-step quantum walk}
Here we slightly generalize the one discussed in Example \ref{ex:SS}.  
Let $p= \sin (\theta_2/2)$ and $q = \cos (\theta_2/2)$
with $\theta_2 \in [-2\pi,2\pi]$. 
Let $\theta_1:\mathbb{Z} \to [0,2\pi)$ be a function
and define $C(x)$ as in \eqref{def_coin} with 
$a(x) = -\sin (\theta_1(x)/2)$ 
and $b(x) = \cos (\theta_1(x)/2)$.
Similarly to the argument in  Example \ref{ex:SS}, 
$U$ is unitarily equiavalent to $U_{\rm ss}(\theta_1(\cdot), \theta_2)$.
In this case, we can take a normalized vector $\chi(x) \in \ker(C(x)-1)$
as
\begin{equation}
\label{chi} 
\chi(x) = 
	\frac{1}{\sqrt{2(1+\sin (\theta_1(x)/2))}}
	\begin{pmatrix} 
	\cos(\theta_1(x)/2) \\
	1 + \sin (\theta_1(x)/2)
	 \end{pmatrix}. 
\end{equation} 
By definition, if $\theta_2, \theta_1(x) \not=  \pi$,
then Hypotheses 1 and 2 are satisfied. 
Suppose that the limits
$\theta_{\pm \infty} 
:= \lim_{x \to \pm \infty} \theta_1(x) \in [0,2 \pi) 
\setminus \{\pi\}$
exist. 
We define $\chi_{\pm \infty}$ as in \eqref{chi} 
with $\theta_1(x)$ replaced by $\theta_{\pm \infty}$. 
By direct calculation, 
\begin{equation}
\label{eq:6.5}
\beta_\pm(-\infty) 
= 
	\frac{1-\sin(\theta_{- \infty}/2)}{1+\sin(\theta_{- \infty}/2)} 
	\frac{1 \mp \sin(\theta_2/2)}
		{1 \pm \sin(\theta_2/2)}
	,
\quad
\beta_\pm(+\infty) 
= 
	\frac{1+\sin(\theta_{+ \infty}/2)}{1-\sin(\theta_{+ \infty}/2)} 
	\frac{1 \pm \sin(\theta_2/2)}
	{1 \mp \sin(\theta_2/2)}. 
\end{equation}
\begin{Theorem}
\label{thm:6.2}
\begin{itemize}
\item[(1)] If
$\sin (\theta_{-\infty}/2) < \sin (\theta_{+\infty}/2)$,
then $\mathcal{B}_\pm = \{0\}$.
\item[(2)] If
$\sin (\theta_{-\infty}/2) > \sin (\theta_{+\infty}/2)$,
then the following hold:
\begin{itemize}
\item[$\bullet$] If $\mp \sin (\theta_2/2) 
\in (\sin (\theta_{+\infty}/2),  \sin (\theta_{-\infty}/2))$,
then $\dim \mathcal{B}_\pm = 1$;
\item[$\bullet$] If $\mp \sin (\theta_2/2) < \sin (\theta_{+\infty}/2)$
or $\sin (\theta_{-\infty}/2)) < \mp \sin (\theta_2/2) $, 
then $\mathcal{B}_\pm = \{0\}$.
\end{itemize}
\end{itemize}
\end{Theorem}
\begin{proof}
From \eqref{eq:6.5}, we obtain the following assertions.
\begin{itemize}
\item[(a)] $\beta_\pm(-\infty) < 1$ if and only if 
$\mp \sin (\theta_2/2) < \sin (\theta_{-\infty}/2)$. 
\item[(b)] $\beta_\pm(+\infty) < 1$ if and only if 
$\sin (\theta_{+\infty}/2)  <  \mp \sin (\theta_2/2)$. 
\end{itemize}
Combining (a) and (b) with Theorem \ref{thm:5.2},
we obtain the desired results. 
\end{proof}
\begin{Remark}
In the case of  (2) in Theorem \ref{thm:6.2},
we observe the following. 
\begin{itemize}
\item If both  $-\sin (\theta_2/2)$ and $+\sin (\theta_2/2)$ are in  
$(\sin (\theta_{+\infty}/2),  \sin (\theta_{-\infty}/2))$,
then 
\[ \dim \mathcal{B}_+ = \dim \mathcal{B}_- = 1. \] 
\item If $\mp \sin (\theta_2/2) \in (\sin (\theta_{+\infty}/2),  \sin (\theta_{-\infty}/2))$ and 
$\pm \sin (\theta_2/2) \not\in [\sin (\theta_{+\infty}/2),  \sin (\theta_{-\infty}/2)]$, then 
\[ \dim \mathcal{B}_\pm = 1, \quad \mathcal{B}_\mp = \{0\}. \]   
\end{itemize}
\end{Remark}
\noindent
\section*{Acknowledgements}
This work was supported by Grant-in-Aid for Young Scientists (B) (No. 26800054).

\end{document}